\documentclass[10pt]{article}
\usepackage{amssymb, amsmath, amsthm, bm, float, graphicx,  tabularx, geometry, mathrsfs, setspace, enumerate, appendix, csquotes, caption, url, yfonts, pslatex, sectsty, eulervm, xcolor, hyperref}
\usepackage[T1]{fontenc}
\usepackage[round]{natbib}
\newtheorem{theorem}{Theorem}

\newtheorem*{definition}{Definition}
\newtheorem{example}{Example}

\newtheorem{lemma}{Lemma}

\newtheorem{proposition}{Proposition}
\newtheorem{question}{Question}
\newtheorem{remark}{Remark}
\newtheorem{result}{Result}

\newcommand{\dom}{\text{dom}}

\newcommand{\even}{\textsc{even}}

\newcommand{\ifif}{\Leftrightarrow}

\newcommand{\odd}{\textsc{odd}}

\newcommand{\ran}{\text{ran}}

\newcommand{\gse}{\text{GE}}

\newcommand{\we}{\text{WE}}
\newcommand{\gpd}{\text{GPD}}

\newcommand{\wpd}{\text{WPD}}
\newcommand{\RR}{\mathbb{R}}

\newcommand{\N}{\mathbb{N}}

\newcommand{\Z}{\mathbb{Z}}
\newcommand{\NM}{\mathbb{N}^{-}}

\input{tcilatex}
\begin{document}

\title{Equitable preference relations on infinite utility streams}
\author{Ram Sewak Dubey\thanks{%
Department of Economics, Feliciano School of Business, Montclair State
University, Montclair, NJ 07043; E-mail: dubeyr@montclair.edu} \and %
Giorgio Laguzzi\thanks{%
University of Freiburg in the Mathematical Logic Group at Eckerstr. 1, 79104
Freiburg im Breisgau, Germany; Email: giorgio.laguzzi@libero.it}}

\date{\today}
\maketitle
\begin{abstract}
We propose generalized versions of strong equity and Pigou-Dalton transfer principle.
We study the existence and the real-valued representation of social welfare relations satisfying  these two generalized equity principles. 
Our results characterize the restrictions on one period utility domains for the equitable social welfare relations (i) to exist; and (ii) to admit real-valued representations.

\noindent \emph{Keywords:} \texttt{Anonymity,}\; \texttt{Generalized equity,}\; \texttt{Generalized Pigou-Dalton transfer principle,}\;  \texttt{Social welfare function,} \;  \texttt{Social welfare relation.} 

\noindent \emph{Journal of Economic Literature} Classification Numbers: \texttt{C65,} \texttt{D63,}\; \texttt{D71.}
\end{abstract}

\newpage

\section{Introduction}\label{s1}

The subject matter of \emph{inter-generational} equity has received considerable attention in economics and philosophy literature in recent times.
It deals with the important question of how to treat the welfare of infinite future generations compared to those living at the present time.
The equity principles employed in this literature can be classified in two broad categories, namely, procedural and consequentialist (or redistributive).

Equity concepts which show indifference between any pair of infinite streams whenever the changes do not alter the distribution of utilities are called procedural equity principles.
\citet{ramsey1928} while devising the idea of inter-generational equity, observed  that discounting one generation's utility relative to another's is \enquote{ethically indefensible}, and something that \enquote{arises merely from the weakness of the imagination}. 
\citet{diamond1965} formalized the concept of \enquote{equal treatment} of all generations (present and future) in the form of an \emph{anonymity} axiom (AN) on social preferences.
Anonymity requires that the society should be indifferent between two infinite streams of well-being, if one is obtained from the other by interchanging the levels of well-being of any two generations.

Equity concepts which require an alteration in the distribution of utilities are called consequentialist equity principles.
Several formulations of this equity notion has been discussed in the social choice literature.
Important redistributive equity axioms are \emph{Pigou-Dalton transfer principle} (PD), \emph{strong equity} (SE) and \emph{Hammond equity} (HE).
Originally, these equity notions were conceptualized in social choice models with finitely many agents or in finite time horizon economies, which consider finitely many generations.
However, these ideas have now been extended to the standard infinite horizon economy with infinitely many generations or in social choice models with infinitely many agents.
The standard framework in the literature treats infinite utility streams as elements of the set $X\equiv Y^{\N}$, where $Y\subset \RR$ is a non-empty subset of real numbers and $\N$ is the set of natural numbers. 
The set $Y$ describes all possible levels of utility that any generation can attain. 
The object of study in our paper is  \emph{social welfare relation} (a reflexive and transitive binary relation) which satisfies some equity axioms.
In case the social welfare relation is complete, it is known as \emph{social welfare order} and a real-valued function representing the social welfare order is termed as a \emph{social welfare function}.
The PD and SE axioms generate strict preference (ranking) for an infinite utility stream obtained from another after carrying out desirable redistribution.
In contrast, HE generates weak preference (ranking) in similar situations.
In what follows, we present a brief description of these axioms in the general setting of comparing infinite utility streams, $x, y\in X$.

Strong Equity compares two infinite utility streams ($x$ and $y$) in which all generations except two have the same utility levels in both utility streams. 
Regarding the two remaining generations (say, $i$ and $j$), one of the generations (say $i$) is better off in utility stream $x$, and the other generation ($j$) is better off in utility stream $y$, thereby leading to a situation of a conflict. 
Strong equity requires that if for both utility streams, it is generation $i$ which is worse off than generation $j$, then generation $i$ should choose (on behalf of the society) between $x$ and $y$. 
Thus, utility stream $x$ is socially preferred to $y$, since generation $i$ is better off in $x$ than in $y$.
In essence, strong equity axiom yields strict preference for inequality reducing transfers which do not alter the ranking of the rich and the poor generations involved in the transfer.
Strong equity was introduced by \citet{aspremont1977}, who referred to it as an \enquote{extremist equity axiom}.
The term \enquote{strong equity} has been used relative to the axiom introduced by \citet{hammond1976}, which he called the \enquote{equity axiom}.
Hammond explains his axiom to be in the spirit of the \enquote{weak equity axiom} of \citet{sen1973}.
His axiom is now referred to as the Hammond equity principle.
In the setting of strong equity ranking, Hammond equity yields weak preference of $x$ compared to $y$.

Pigou-Dalton transfer principle insists on exact transfer of the utility (welfare) from the rich generation to the worse off generation ($x_i-y_i=y_j-x_j$) for the desired social preference in the redistribution scheme described above for strong equity.
This inequality reducing property was initially hinted at by \citet[p. 24]{pigou1912} as \enquote{The Principle of Transfers}.
\citet[p. 351]{dalton1920} described it as \enquote{If there are only two income-receivers and a transfer of income takes place from the richer to the poorer, inequality is diminished.}.
In words, Pigou-Dalton transfer principle shows preference for non-leaky and non-rank-switching transfers.
This equity concept has been widely studied in literature on economic inequality and philosophy (see \citet{adler2013} for a discussion on philosophical foundation).%
\footnote{There are numerous variations of redistributive equity principles discussed in social choice literature.
Prominent examples include \enquote{proportional transfer principle and proportional ex post transfer principle } (\citet{fleurbaey2001}), \enquote{distributive fairness semiconvexity and strong distributive fairness semiconvexity} (\citet{sakai2003}),  \enquote{Hammond equity for the future} (\citet{asheim2007}, \citet{banerjee2006}), \enquote{altruistic equity (altruistic equity-1 (AE-1) and altruistic equity-2 (AE-2)} (\citet{hara2008} and \citet{sakamoto2012}), \enquote{weak Pigou-Dalton transfer principle} (\citet{sakai2010}), \enquote{very weak inequality aversion} (\citet{alcantud2012}), and \enquote{pairwise Hammond equity} (\citet{sakai2016}).
Pigou-Dalton transfer principle has also been discussed in the inter-generational equity literature as \enquote{strict transfer principle} by \citet{sakai2006}, \citet{bossert2007}, and \citet{hara2008}.}

The literature on inter-generational equity deals with improving the welfare of current and future generations extending over an infinite time horizon.
For this, we need to evaluate infinite utility streams consistently with social preferences which respect the desirable class of equity axiom(s).
In what follows, we will refer to the social preferences satisfying equity axiom(s) as ethical or equitable preferences.
In the case of procedural equity notion of anonymity, it is easy to construct an ethical social preference.
A social preference relation which evaluates all infinite utility streams as indifferent satisfies the anonymity axiom (in addition to Hammond equity) trivially. 
However such a preference relation is of limited practical significance as it does not distinguish between any pair of utility sequences.
Since anonymity and Hammond equity axioms fail to generate any strict preference, additional conditions in the nature of sensitivity (Pareto)%
\footnote{The Pareto axiom requires a stream of utilities to be ranked above another one if at least one  generation is better off and no one is worse off in the former compared to the latter.} 
and/or continuity axioms are imposed to obtain actionable social preferences.
In contrast, social preference relations satisfying strong equity or Pigou-Dalton transfer principle enable strict ranking of pairs of utility sequences.
Therefore, properties of such preferences without any additional sensitivity (Pareto) or continuity axioms are worthy of investigation.
Our goal is to propose a general framework for this study and report results of our inquiry.
We begin with some more historical background and a brief overview of the recent literature on the inter-generational equity and discuss in detail the scope and position of our contribution within the established literature.

\subsection{More on historical background and recent developments}

One of the classic papers on the study of Pigou-Dalton transfer principle is \citet{atkinson1970} which deals with the subject matter of inequality measures for finite population samples.
While analyzing commonly used inequality indices (variance, coefficient of variation or Gini index), Atkinson noted Dalton's views on the need for studying the underlying concept of social welfare.
Atkinson utilized the class of social welfare functions which are additively separable and symmetric functions of individual incomes ($U(y)$, where $y$ is the income of any individual) to understand their relationships with the inequality measures.
He showed that the necessary and sufficient conditions for the social welfare function to satisfy Pigou-Dalton transfer principle is that $U(y)$ is increasing and concave. 
\citet{dasgupta1973} refined the conclusions in \citet{atkinson1970} and showed that any symmetric social welfare function which is increasing and quasi-concave in individual incomes would satisfy the transfer principle.
In particular, for the social welfare function satisfying Pigou-Dalton transfer principle, additive separability is not an essential characteristic.
Their proof (\citet[Lemma 2, and Theorem 1]{dasgupta1973}) relied on an inequality established by Hardy, Littlewood and Polya.
\citet{fleurbaey2001} observed that the transfer principle has been criticized because of the costless transfer of income (or consumption) from the rich to the poor as not  a realistic condition (i.e., there is some loss of welfare in transmission from the rich to the poor).
They consider variations of the Pigou-Dalton transfer principle, namely proportional transfer principle and proportional ex-post transfer principle and show that the social welfare functions satisfying these equity conditions must exhibit forms of concavity stronger than log-concavity.
It is pertinent to note that all of the aforementioned research papers consider application of Pigou-Dalton transfer principle (and its variations) to a society with finite population only.

In general, the real-valued representation satisfying PD does not carry over from the finite time horizon (finitely many agents) economy to infinite time horizon (infinitely many agents) models.
To overcome this difficulty in consistent ranking of infinite utility streams, recourse is taken to the analysis of pairwise rankings via social welfare relations and social welfare orders whenever it is feasible to do so.
There is a rich body of research publications dealing with evaluations of infinite utility streams satisfying redistributive equity and other useful monotonicity and/or continuity axioms.
\citet[Section 5]{hara2008} have categorized the pool of possibility/impossibility results on consistent evaluations of infinite utility streams based on three criteria: (a) rationality (transitivity, negative transitivity, quasi-transitivity, acyclicity, completeness, etc.); (b) continuity (or upper semi-continuity, lower semi-continuity etc.);  and (c) sensitivity (variations of Pareto, procedural equity (AN), redistributive equity (PD, SE, HE).
The results described below could be clubbed together based on these three conditions.%
\footnote{We refer the interested reader to \citet{hara2008} for further details.}

\citet{sakai2003} extended the redistributive equity axioms to the setting of infinite utility streams (i.e., infinite population).
\citet{sakai2003} introduced two versions of redistributive equity condition, namely, distributive fairness semiconvexity and strong distributive fairness semiconvexity for pairwise ranking on the Banach space of all bounded sequences $X\equiv \ell^{\infty}_{+}:= \left\{\langle x_t\rangle_{t=1}^{\infty}: \sup_{t} |x_t|<\infty \right\}$.%
\footnote{Though the strong distributive fairness semiconvexity appears to be similar in spirit to the Pigou-Dalton transfer principle, it is distinct and more restrictive.
For illustration, consider $x=(1, 0, 0, 1, 0, 0, \cdots)$, $x(\pi)= (0, 1, 1, 0, 0, 0, \cdots)$ and $y = (0.7, 0.3, 0.4, 0.6, 0, 0, \cdots)$. 
Then there is no $s\in(0, 1)$ such that $y = sx+(1-s)x(\pi)$.
Hence, strong distributive fairness semiconvexity is not applicable whereas $y\succ x$ by Pigou-Dalton transfer principle and $y\succ x(\pi)$ applying anonymity and  Pigou-Dalton transfer principle.}
He proved the existence of an anonymous Paretian social welfare order satisfying strong distributive fairness semiconvexity (a positive result) and the non-existence of any anonymous social welfare relation satisfying distributive fairness semiconvexity and sup-norm continuity axioms (a negative result).
The negative result points to the difficulty one faces in seeking any (sup-norm) continuous anonymous social welfare order satisfying redistributive equity conditions.
\citet{sakai2006} considered Pigou-Dalton transfer principle in the infinite horizon setting (i.e., on the space of utility streams $X\equiv \ell^{\infty}$) and has  established that under very reasonable restrictions on the utility domain $X$, there does not exist any (lower or upper) continuous binary relation satisfying Pigou-Dalton transfer principle and quasi-transitivity.
\citet[Theorems 1 and 2]{sakai2006} clearly show the need for relaxing some of the axioms (quasi-transitivity and/or upper (or lower) continuity) to obtain social welfare preferences satisfying Pigou-Dalton transfer principle.
\citet[Theorems 1 and 2]{bossert2007} proved the existence of an anonymous Paretian social welfare order satisfying (a) Pigou-Dalton transfer principle (Theorem 1); and (b) strong equity (Theorem 2) on the space of infinite utility streams, $X=\RR^{\N}$.
\citet[Theorem 1]{hara2008} have proven that there does not exist any acyclic social evaluation relation satisfying Pigou-Dalton transfer principle and a weak version of continuity (with respect to sup-metric).
\citet[Theorem 2, Corollary 3]{sakai2010} established the non-existence of binary relations satisfying weak Pigou-Dalton transfer principle and versions of anonymity and transitivity.
\citet[Corollary 2]{sakai2016} considered pairwise Hammond equity axiom in combination with a set of four additional axioms (of efficiency, and anonymity) and proved that the unique real-valued representation of the ethical social welfare orders is the $\liminf$ function.

One can infer from the short description above that there is a need to dispense with some of the axioms in order to focus on social welfare relations which respect the Pigou-Dalton transfer principle.
Typical objective of such an exercise is to investigate the existence of ethical social welfare functions satisfying redistributive equity conditions (PD, SE or HE).
There are positive as well as negative results available in the literature. 
\citet{alcantud2012} proved the non-existence of social welfare functions satisfying Pigou-Dalton transfer principle and weak Pareto%
\footnote{The weak Pareto axiom requires a stream of utilities to be ranked above another one if every generation is better off in the former compared to the latter.}
for $X=\RR^{\N}$.
\citet[Proposition 5]{alcantud2010a} and \citet[Proposition 3]{sakamoto2012}  proved the following possibility result. 
There exists a social welfare function satisfying Pigou-Dalton transfer principle, anonymity and weak dominance (an efficiency condition) for $X=[0, 1]^{\N}$.
\citet{sakamoto2012} also proved other possibility and impossibility results dealing with variations of strong equity (AE-1 and AE-2).
\citet[Corollary 1]{alcantud2013a} proved the non-existence of any social welfare function satisfying strong equity and monotonicity for $X=[0, 1]^{\N}$.
\citet[Theorem 1]{alcantud2013} proved the non-existence of any social welfare function satisfying Hammond  equity and Pareto  for $X=Y^{\N}$, with $|Y|\geq 4$ and the existence of a social welfare function satisfying Hammond equity and weak Pareto (Theorem 2) when $Y=\N$.%
\footnote{\citet{dubey2015} extend their analysis and prove the following: There exists a social welfare function satisfying HE and weak Pareto if and only if the set $Y$ is well ordered.}
\citet{dubey2014} and \citet{dubey2016} also prove several positive and negative results on the existence of social welfare functions satisfying strong equity and Pigou-Dalton transfer principle, respectively.%
\footnote{See also \citet[Propositions 1 and 2]{alcantud2020}.}

This brief overview reveals that the following feature of the redistributive equity condition is present in most of the analysis (with the exception of weak Pigou-Dalton transfer principle of \citet{sakai2010} and pairwise Hammond equity of \citet{sakai2016}).
The definitions of the redistributive equity axioms can account for at most finitely many equity enhancing transfers.
This is in contrast to the efficiency condition of Pareto axiom which allow arbitrarily many welfare improvements. 
We believe the restrictive feature of the redistributive equity axioms need to be examined in order to enrich our knowledge of the ethical social welfare relations. 
The next sub-section describes the generalization proposed in this paper.

\subsection{Generalizing the distributional equity principles}

Careful scrutiny of the equity conditions mentioned earlier shows that in general the inequality reducing transfers are restricted to  finitely many generations in utility stream $x$ to get the preferred utility stream $y$. 
Two exceptions are \citet{sakai2010} and \citet{sakai2016}.
In the case of weak Pigou-Dalton transfer principle of \citet{sakai2010} and pairwise Hammond equity of \citet{sakai2016},  redistribution involving infinitely many generations are considered, but only of a specific type, which is to carry out a redistribution involving all generations and consisting of a welfare redistribution between a generation and the immediate subsequent one.

However, more general transfer mechanisms (for instance among an arbitrary infinite set of generations) have not been considered in the literature.
It is pertinent to note that the transitivity property of strict ranking need not carry over from finite to infinite case.
The following example describes this situation. 
Consider the utility stream 
\[
\begin{matrix}
x=&\langle&1,& 0,& 1,&1,& 0,& 1,& 1,& 0,&1,&1,&0,&\cdots\rangle\\
z=&\langle&0.75,& 0.25,& 1,&0.6,& 0.1,& 1,&1,& 0,&1,&1,&0,&\cdots\rangle,
\end{matrix}
\]
where (a) generations $2$, $5$, $8$, $\cdots$ in $x$ have utility level zero and the remaining generations have utility level one, and (b) $z$ is obtained from $x$ by carrying out transfers for the first two generations of rich and poor agents, and generation four (rich) and five (poor)  without changing the rest of the generations.
Strong equity enables us to prefer $z$ to $x$.
However, if we consider 
\[
\begin{matrix}
z^{\prime} = \langle0.75,& 0.25,& 1,&0.6,& 0.1,& 1,&0.6,& 0.1,&1,&0.6,&0.1,&1,&0.6,&0.1,&\cdots\rangle,
\end{matrix}
\]
neither strong equity nor weak Pigou-Dalton transfer principle are helpful in comparing $z^{\prime}$ (obtained by carrying out infinitely many inequality diminishing transfers from $(1, 0)$ to $(0.6, 0.1)$) and $z$, even though intuitively it makes sense to  prefer $z^{\prime}$ in comparison to $z$.

The main contribution in this paper is to introduce some generalized variants of strong equity and (weak or standard) Pigou-Dalton transfer principle in order to be able to consistently rank pairs of utility streams like $z^{\prime}$ and $z$ mentioned above.
The new equity concepts are termed as \emph{generalized equity} (GE) and \emph{generalized Pigou-Dalton transfer principle} (GPD).
We report new and interesting features of generalized equity in this comprehensive study. 
Our results show that the generalized equity principles behave rather differently from the standard strong equity and Pigou-Dalton transfer principle. 
Unlike the case for strong equity, it is not possible to find a binary relation satisfying generalized equity when $Y$ contains the interval $[0,1]$.
The existence of social welfare relations is determined by the \text{order type} of the set $Y$.
In Examples \ref{Ex1} and \ref{Ex2} we show the restriction on the structure of the utility domain $Y$ which should be imposed in order to obtain the existence of social welfare relations satisfying generalized equity (and monotonicity): 
Lemma \ref{L3} provides a sufficient condition and shows that when $Y$ is well-ordered then social welfare relations satisfying generalized equity and monotonicity exist. 
Since the characterizing feature is the order property of the set $Y$, this result is not sensitive to the numerical values of the elements of $Y$.%
\footnote{For example, $Y:=\left\{\frac{1}{n+2}, \frac{n}{n+1}:n\in\N\right\}$ and $Y^{\prime}:=\{-1, 1,-2, 2,\cdots\}$ are such that $Y(<)$ and $Y^{\prime}(<)$ have identical order type, i.e.,  $Y(<)$ is order isomorphic to the set of integers.
On the other hand, $\textbf{Y}:\left\{\frac{1}{2}-\frac{1}{n+1}, \frac{1}{2}+\frac{1}{n+1}:n\in\N\right\}$ is such that $\textbf{Y}(<)$ and $Y(<)$ are not order isomorphic to each other. 
Set $\textbf{Y}$ is equipped with well-defined minimum, $0$, and maximum, $1$, whereas $Y$ has neither.}
Moreover, Lemma \ref{L1} ensures that the existence result in Lemma \ref{L3} is invariant to monotone (increasing/decreasing) transformation of the elements in the domain set $Y$. As a consequence we can simply restrict our attention to utility domain $Y \subseteq [0,1]$. 

Beyond generalized equity (GE) we also introduce other two variants, in line with the literature on Pareto principles for infinite utility streams, which are called \emph{infinite equity} (IE) requiring the welfare redistribution to operate on an arbitrary infinite set of generations, and \emph{weak equity} (WE) requiring the redistribution to involve all generations, with the rich and poor generations being paired (linked) in any arbitrary manner (see precise definition in section \ref{s3}).   

Having established sufficient conditions  for the existence of social welfare relations, we next examine their real-valued representation via social welfare functions. 
In section \ref{s4}, Proposition \ref{P1} contains an example of a social welfare function satisfying generalized equity and monotonicity when $Y$ contains no more than five elements.
Proposition \ref{P2} then shows that $Y$ could be enlarged to no more than seven elements if we dispense with the efficiency condition.

The main results of the paper are Theorems \ref{T1}, \ref{T2}, \ref{T3}, \ref{T4} and \ref{T5}, which consist of five characterizations of real-valued representation via social welfare functions for various combinations of GE, IE, WE, together with monotonicity (M) and anonymity (AN).
Theorem \ref{T1} shows the impossibility of real-valued representation for infinite equity combined with anonymity as soon as  $Y$ contains four elements (the least number of elements needed to employ generalized equity).
This points out inherent conflict between representability, infinite equity and anonymity for all non-trivial domains $Y$. 
Theorem \ref{T1} also answers negatively to a natural question related to Propositions \ref{P1} and \ref{P2} about the existence of social welfare functions satisfying generalized equity and anonymity.%
\footnote{\citet{fleurbaey2003} have examined the difficulties in the explicit description of  social welfare orders satisfying anonymity and efficiency axioms. 
Their analysis is relevant in view of the non-existence of social welfare functions proven in Theorem \ref{T1}.}
Theorem \ref{T2} shows that there exists a social welfare function satisfying infinite equity and monotonicity if and only if $Y$ consists of at most five elements. 
Theorem \ref{T3} shows that there exists a social welfare function satisfying infinite equity if and only if $Y$ consists of at most seven elements.
In particular, Theorems \ref{T2} and \ref{T3} show that the number of utility levels in $Y$ occurring in Proposition \ref{P1} and Proposition \ref{P2} are optimal, since in both cases the possibility of real-valued representations does not extend to larger utility domains.

Due to this persistent conflict between infinite equity and representation, we consider the still weaker equity  
notion of weak equity.
As reported earlier in Example \ref{Ex1}, consistent pair-wise ranking of utility streams satisfying weak equity is not possible in case $Y$  is any arbitrary subset of real numbers.
In Theorem \ref{T4}, we show that social welfare relations satisfying weak equity, monotonicity (and anonymity) admit representation if and only if $Y$ is well-ordered.
In Theorem \ref{T5}, we show that social welfare relations satisfying weak equity (and anonymity) admit representation if and only if $Y$ does not contain a subset order isomorphic to the set of integers (both natural numbers and negative integers).
It is also worthy of mention (Remark \ref{R3}) that the characterization results in Theorems \ref{T4} and \ref{T5} show that the representation of weak equity (both when combined with monotonicity and alone) is not effected by anonymity. 
This underlines a substantial difference compared to the representability results obtained for infinite equity (and monotonicity) combined with anonymity, as proven in Theorems \ref{T1}, \ref{T2} and \ref{T3} above mentioned, where anonymity is shown to make impossible real-valued representation for any non-trivial utility domain.

The paper is organized as follows.
We introduce notation and basic definitions in section \ref{s2}.
Sections \ref{s3} and \ref{s4} form the core of the paper.
In section \ref{s3} we introduce the generalized equity principles extending strong equity and Pigou-Dalton principles, and then study conditions for the existence of the social welfare relations satisfying these generalized equity principles. 
Section \ref{s4} deals with the real-valued representation of the various versions of generalized equity.
We conclude in section \ref{s5}.
All proofs (except proof of Theorem \ref{T1}) are in the appendix.

\setcounter{theorem}{0}

\section{Preliminaries}\label{s2}
\subsection{Notation}\label{s21}
Let $\RR$, $\N$, $\NM$, $\Z$ be the sets of real numbers, natural numbers, negative integers and integers respectively. 
For all $y, z\,\in \RR^{\N}\,$, we write $y\geq z$ if $y_{n}\geq z_{n}$, for all $n\in \N$; we write $y>z$ if $y\geq z$ and $y\neq z;$ and we write $y\gg z$ if $y_{n}>z_{n}$ for all $n\in \N$.
For a set $S$, we denote its cardinality by $|S|$.
Recall that a subset of $\N$ is called \emph{co-finite} if it is the complement of a finite set. 

\subsection{Definitions}\label{s22}

\subsubsection{Pairing functions}\label{ss221}
A \emph{partial function} $f: X \rightarrow Y$ is a function from a subset $S$ of $X$ to $Y$. 
If $S$ equals $X$, the partial function is said to be \emph{total}.
Domain and range of function $f$ are denoted by $\dom(f)$ and $\ran(f)$ respectively.
We say that $\alpha: \N \rightarrow \N$ is a \emph{pairing function} if and only if $\alpha$ is a partial function satisfying:
\begin{itemize}
\item $\dom(\alpha)=\ran(\alpha)$
\item $\forall n \in \dom(\alpha)$,  $\alpha(\alpha(n))=n$.  
\end{itemize}
We denote the set of all pairing functions by $\Pi$.

\subsubsection{Ultrafilter on the set of natural numbers}\label{ss222}
A collection $\mathcal{F}$ consisting of subsets of $\N$ is called a \emph{filter} on $\N$ if and only if it satisfies the following properties:
\begin{itemize}
\item $\N \in \mathcal{F}$ and $\emptyset \notin \mathcal{F}$,
\item $\forall A$, $B \in \mathcal{F}$, $A \cap B \in \mathcal{F}$,
\item $\forall A \in \mathcal{F}$,  $\forall B \subseteq \N$ ($B \supseteq A \Rightarrow B \in \mathcal{F}$).
\end{itemize}
Moreover if $\mathcal{F}$ is maximal (i.e. for every filter $\mathcal{F}^{\prime} \supseteq \mathcal{F}$ we have $\mathcal{F}^{\prime}=\mathcal{F}$), then $\mathcal{F}$ is said to be an \emph{ultrafilter}.

\subsubsection{Order properties of subsets of real numbers} \label{ss223}
We will say that the set $S$ is \emph{strictly ordered} by a binary relation $\Re$ if $\Re$ is \emph{connected} (if $s$, $s^{\prime}\in S$ and $s\neq s^{\prime}$, then either $s\Re s^{\prime}$ or $s^{\prime}\Re s$ holds), \emph{transitive} (if $s$, $s^{\prime}$, $s^{\prime\prime}\in S$ and $s\Re s^{\prime}$ and $s^{\prime}\Re s^{\prime\prime}$ hold, then $s\Re s^{\prime\prime}$ holds) and \emph{irreflexive} ($s\Re s$ holds for no $s\in S$). 
In this case, the strictly ordered set will be denoted by $S(\Re)$.
For example, the set $\N$ is strictly ordered by the binary relation $<$ (where $<$ denotes the usual \enquote{less than} relation on the real numbers).
We will say that a strictly ordered set $S^{\prime}(\Re^{\prime})$ is \emph{order isomorphic} to the strictly ordered set $S(\Re)$ if there is a one-to-one function $f$ mapping $S$ onto $S^{\prime}$, such that for all $s_{1}$, $s_{2}\in S$, $s_{1}\Re s_{2}$ implies $f(s_{1})\Re^{\prime}f(s_{2})$. 
We next consider strictly ordered subsets of $\RR$. 
With non-empty $S\subset \RR$, by stardard notation, the strictly ordered set $S(<)$ is called:
\begin{enumerate}
\item[(i)]{of order type $\omega$ if $S(<)$ is order isomorphic to $\N(<)$;}
\item[(ii)]{of order type $\omega^{\ast}$ if $S(<)$ is order isomorphic to $\NM(<)$;}
\item[(iii)]{of order type $\sigma$ if $S(<)$ is order isomorphic to $\Z(<)$.}
\end{enumerate}
In order to characterize these types of strictly ordered sets, we need to define the concepts of a \emph{cut}, a \emph{first} element and a \emph{last} element of a strictly ordered set.
Given a strictly ordered set $S(<)$, we define a \emph{cut} $[S_{1}, S_{2}]$ of $S(<)$ as a partition of $S$ into two sets $S_{1}$ and $S_{2}$ (i.e., $S_{1}$ and $S_{2}$ are non-empty, $S_{1}\cup S_{2}=S$ and $S_{1}\cap S_{2}=\emptyset$) such that for each $s_{1}\in S_{1}$ and each $s_{2}\in S_{2},$ we have $s_{1}<s_{2}$.
An element $\underline{s}\in S$ ($\overline{s}\in S$) is called a first (last) element of $S(<)$ if $s<\underline{s}$ ($\overline{s}<s$) holds for no $s\in S$. 
A strictly ordered set $S(<)$ is said to be \emph{well-ordered} if each non-empty $A\subset S$ has a first element (namely $\min A$). 
It follows from the above definitions that if the strictly ordered sets $S(<)$ and $T(<)$ are order isomorphic, then $S(<)$ is well-ordered if and only if $T(<)$ is well-ordered.
A basic characterization result on well-ordered sets is as follows (see \citet[Theorem 1, p. 262]{sierpinski1965}).

\begin{result}\label{RE1}
\emph{Let $S$ be a non-empty subset of $\mathbb{R}$ and $S(<)$ a strictly ordered set. 
The necessary and sufficient condition for $S(<)$ to be a well-ordered set is that it should contain no subset of order type $\omega^{\ast}$; i.e., it should contain no infinite strictly decreasing sequence.}
\end{result}
\noindent The following result on the strictly ordered set $S(<)$ of type $\sigma$ can be found in \citet[p. 210]{sierpinski1965}.
\begin{result}\label{RE2}
\emph{A strictly ordered set $S(<)$ is of order type $\sigma$ if and only if the following two conditions hold:
\begin{enumerate}
\item[(i)]{$S$ has neither a first element nor a last element, or in other words neither $\min S$ nor $\max S$ exists.}
\item[(ii)]{For every cut $[S_{1},S_{2}]$ of $S$, the set $S_{1}$ has a last element and the set $S_{2}$ has a first element.}
\end{enumerate}}
\end{result}

\subsubsection{Social welfare relations}\label{ss224}
Let $Y$, a non-empty subset of $\RR$, be the set of all possible utilities that any generation can achieve. 
Then $X\equiv Y^{\N}$ is the set of all possible utility streams. 
If $\langle x_{n}\rangle \;\in \;X$, then $\langle x_{n}\rangle = \langle x_{1}, x_{2}, \cdots\rangle$, where, for all $n\in \N$, $x_{n}\in Y$ represents the amount of utility that the generation of period $n$ earns.
For $x\in X$, we use notation $x[n]:=\langle x_1, x_2, \cdots, x_n\rangle$ for the initial $n$ terms.
We consider binary relations on $X$, denoted by $\succsim $, with symmetric and asymmetric parts denoted by $\sim $ and $\succ $ respectively, defined in the usual way. 
A \emph{social welfare relation} (SWR) is a reflexive and transitive binary relation.
A \emph{social welfare order} (SWO) is a complete SWR.
A \emph{social welfare function} (SWF) is a mapping $W:X\rightarrow \RR$. 
Given a SWO $\succsim$ on $X$, we say that $\succsim $ can be \emph{represented} by a real-valued function if there is a mapping $W:X\rightarrow \RR$ such that for all $x, y\in X$, we have $x\succsim y$ if and only if $W(x)\geq W(y)$.

\subsubsection{Equity and efficiency axioms}\label{ss225}

In this paper, we study the generalization of the following consequentialist equity criteria.
\begin{definition}
\emph{(Pigou-Dalton transfer principle - PD): \quad If $x,y\in X$, and there exist $i,j\in \N$ and $\varepsilon>0$, such that $x_{j} - \varepsilon = y_{j} > y_{i} = x_{i}+\varepsilon$, while $y_{k} = x_{k}$ for all $k\in \N\setminus \{i,j\}$, then $y\succ x$.}
\end{definition}
\begin{definition}
\emph{(Strong equity - SE):\quad If $x,y\in X$, and there exist $i,j\in \N$, such that $x_{j} > y_{j} > y_{i} > x_{i}$ while $y_{k} = x_{k}$ for all $k\in \N\setminus \{i,j\}$, then $y\succ x$.}
\end{definition}

The procedural equity criterion that we will use is anonymity (also known as finite anonymity).
\begin{definition}
\emph{(Anonymity - AN):\quad If $x,y\;\in \;X$, and if there exist $i,j\;\in \;\N$ such that $x_{i} = y_{j}$ and $x_{j} = y_{i}$, and for every $k\in \N\setminus \{i,j\}$, $x_{k} = y_{k}$, then $x\sim y$.}
\end{definition}

The following efficiency principle is used throughout the paper.
\begin{definition} 
\emph{(Monotonicity - M):\; Given $x, y\in X$, if $x\geq y$, then $x \succsim y$.}
\end{definition}

\begin{remark}\label{R1}
\emph{The main purpose of this paper is to generalize the consequentialist equity principles such as PD and SE.
Our investigations reveal that there exist significant challenges in consistent evaluation of the social preferences satisfying the proposed generalizations of PD and SE. 
Therefore, when it is feasible to combine additional equity or efficiency conditions, we select the least demanding variants of procedural equity and efficiency axioms.
Monotonicity is a very weak efficiency condition. 
It only requires society to not prefer infinite utility stream $x$ over $y$ if $y_t>x_t$ for some $t\in\N$.
Anonymity (finite) permits ranking of infinite utility streams $x$ and $y$ when one is obtained by carrying out a finite permutation of the other.
There is no conflict in real-valued representation of SWRs satisfying AN and/or M.
As it turns out, significant conflicts already arise when we require real-valued representation of the social preferences satisfying the generalized versions of PD or SE with AN and/or M. Hence, combining our generalized versions of consequentialist equity principles with even stronger versions of efficiency (Pareto, weak Pareto) and/or anonymity would \emph{a fortiori} provide conflicts, in most cases.}  
\end{remark}

\section{Generalized equity and generalized Pigou-Dalton  transfer principles}\label{s3}

The main contribution of this paper in the introduction of a novel concept, named as \emph{generalized equity} (and the related \emph{generalized Pigou-Dalton transfer principle}).
To better understand the key reason for studying such extended versions of the consequentialist equity principles we start with the following observation.
Given a set of utilities $Y:= \{a, b, c, d, e\}\subset\RR$ ordered by $a<b<c<d<e$, consider the two infinite streams $x$ and $y$ such that 
\[
\begin{matrix}
x:=& \langle b,& c,& e,& b,& c,& e,& b,& c,& \cdots&\rangle,\\
y:=& \langle a,& d,& e,& a,& d,& e,& a,& d,& \cdots& \rangle.
\end{matrix}
\] 
Following the expected interpretation of a distributive equity principle, we should be able to always rank $y \prec x$. 
In the finite case, PD (and so SE) together with transitivity is sufficient to secure such a ranking, but in the infinite case transitivity cannot be extended to infinite chains. 
Hence PD and SE even with transitivity are not sufficient conditions to secure the desired ranking $y \prec x$, and so an extension is necessary. 
In this paper, we propose the following generalization of SE and PD axioms and study their characteristics.

\begin{definition}
\emph{(Generalized Equity - GE) Given $x, y \in X$ if there is $\alpha \in \Pi$ such that for every $i \in \dom(\alpha)$ one has either $y_i < x_i < x_{\alpha(i)} < y_{\alpha(i)}$ or $y_{\alpha(i)} < x_{\alpha(i)} < x_i < y_i$ and for all $k \notin \dom(\alpha)$, one has $x_{k} = y_{k}$, then $y \prec x$.}
\end{definition}

\begin{definition}
\emph{(Generalized Pigou-Dalton transfer principle - GPD) Given $x, y \in X$ if there is $\alpha \in \Pi$ such that for every $i \in \dom(\alpha)$ there exists $\varepsilon_i>0$ such that $\text{ either } y_i +\varepsilon_i = x_i < x_{\alpha(i)} = y_{\alpha(i)} -\varepsilon_i \text{ or } y_{\alpha(i)} +\varepsilon_i = x_{\alpha(i)} < x_i = y_i - \varepsilon_i$, and for all $k \notin \dom(\alpha)$, one has  $x_{k} = y_{k}$, then $y \prec x$.}
\end{definition}

Moreover, since we are considering an infinite time horizon, it could be reasonable to require the ranking induced by $\prec$ being sensitive not only to few changes, but to a number of changes as large as possible. 
For instance, in an infinite setting one could require that the number of individuals/generations linked via the pairing function (where we can appreciate a reduction of inequality) should be any arbitrary infinite set. 
This is in line with the weaker forms of Pareto principles extensively studied in the literature, such as infinite Pareto, asymptotic Pareto and weak Pareto.%
\footnote{See \citet{diamond1965}, \citet{basu2003}, \citet{fleurbaey2003}, \citet{zame2007}, \citet{crespo2009}, \citet{lauwers2010}, \citet{dubey2011},  \citet{petri2019}, \citet{dubey2020} and \citet{dubey2021} for versions of Pareto principles employed in inter-generational equity literature.}
The following definitions capture this idea for our study.

\begin{definition} 
\emph{(Infinite Equity - IE). Given $x$, $y \in X$ if there is $\alpha \in \Pi$ such that $|\dom(\alpha)|=\infty$ and for every $i \in \dom(\alpha)$ one has
either $y_i < x_i < x_{\alpha(i)} < y_{\alpha(i)}$ or $y_{\alpha(i)} < x_{\alpha(i)} < x_i < y_i$ and for all $k \notin \dom(\alpha)$ one has $x_k = y_k$, then $y \prec x$.}
\end{definition}

\begin{definition}
\emph{(Weak Equity - WE). Given $x$, $y \in X$ if there is $\alpha \in \Pi$ such that $\dom(\alpha)=\N$ and for every $i \in \dom(\alpha)$ one has
either $y_i < x_i < x_{\alpha(i)} < y_{\alpha(i)}$ or $y_{\alpha(i)} < x_{\alpha(i)} < x_i < y_i$, then $y \prec x$.}
\end{definition}

Note that the analogous notions of infinite and weak Pigou-Dalton transfer principle could be defined in similar manner.    
Moreover, we have a well-ordered ranking of these equity principles (similar to the ranking of Pareto principles) in the following sense:
\[
\gse \Rightarrow  \text{IE} \Rightarrow  \we \quad \text{and} \quad \gpd \Rightarrow  \text{IPD} \Rightarrow \wpd.
\] 
Also, given a certain type of equity principle (e.g. GE, IE, WE), the related Pigou-Dalton version is weaker, which means for instance $\gse \Rightarrow \gpd$ and so on.

\subsection{Social welfare relations satisfying generalized equity} \label{ss31}

\citet{bossert2007} have shown (when $Y=\RR$) that there exist SWOs on infinite utility streams which satisfy Hammond equity, anonymity and Pareto%
\footnote{Pareto axiom is defined as follows: For $x, y\in X$, if $x > y$ then $x\succ y$.} 
axioms.
Under Pareto, Hammond Equity and SE are equivalent. 
Therefore, we infer that  there exist SWOs on infinite utility streams satisfying SE, AN and M.
Similar result is also noted in \citet{lauwers2010}, who uses the notion of a filter in order to obtain an extended lexicographic order to infinite streams satisfying Pareto, AN, PD (and Hammond, SE) axioms.
Further, by using an ultrafilter one could extend such SWR to obtain SWO satisfying the same axioms.
Given these results, it is natural to ask if it is possible to use similar argument to obtain analogous results about the existence of SWRs satisfying GE and GPD. 
However, the situation in the case of GE and GPD is more subtle and the SWRs do not exist in general.
Therefore, some restrictions on the utility domain need to be imposed. The following examples explain this situation.
\begin{example} \label{Ex1}
\emph{There is no SWR satisfying GE on $\mathbb{Z}^\N$. 
Towards a contradiction, assume there exists $\preceq$ on $\Z^\N$ satisfying GE and consider the following streams:
\[
\begin{matrix}
x:=& \langle& 1,& 4,& -1,& -4,& 5,& 8,& -5,& -8,& \cdots,& 4k+1,& 4k+4,& -4k-1,& -4k-4,& \cdots \rangle \\
y:=& \langle& 2,& 3,& -2,& -3,& 6,& 7,& -6&, -7,& \cdots,& 4k+2,& 4k+3,& -4k-2,& -4k-3,& \cdots \rangle.
\end{matrix}
\]
Then let $\alpha$, $\beta \in \Pi$ be defined as follows:
\[
\alpha(n)= \left\{ 
\begin{array}{ll}
n+1 & \text{if $n$ is odd} \\
n-1 &  \text{if $n$ is even} \\
\end{array}
\right.
\quad
\beta(n)= \left\{ 
\begin{array}{ll}
3 & \text{if $n=1$} \\
1 & \text{if $n=3$} \\
n+3 & \text{if $n$ is even} \\
n-3 & \text{if $n$ is odd and $n \geq 5$}
\end{array}
\right.
\] 
A straightforward computation shows that $\alpha$ yields to $x \prec y$, whereas $\beta$ yields to $y \prec x$, providing us with a contradiction. 
Note that the construction actually shows that even WE is inconsistent on $\Z^\N$.}  
\end{example}

\begin{example} \label{Ex2}
\emph{Let $Y=\NM$.
Then GE and M are inconsistent on $Y^\N$.
Towards a contradiction, assume there exists $\preceq$ on $\mathbb{R}^\N$ satisfying GE and M; consider the following streams:
\[
\begin{matrix}
x:=& \langle& -2,& -5,& -6,& -9,& \cdots,& -4k-2,& -4k-5,& \cdots &\rangle \\
y:=& \langle& -3,& -4,& -7,& -8,& \cdots,& -4k-3,& -4k-4,& \cdots& \rangle \\
y^{\prime}:=& \langle& -1,& -4,& -7,& -8,& \cdots,& -4k-3,& -4k-4,& \cdots &\rangle.
\end{matrix}
\]
The sequences consists of elements from the set of negative integers. 
Since $y^{\prime}_1 > x_1>y_1$ and $y^{\prime}_t=y_t$ for all $t>1$, by M, $y \preceq y^{\prime}$.
Then let $\alpha, \beta \in \Pi$ be defined as follows:
\[
\alpha(n)= \left\{ 
\begin{array}{ll}
n+1 & \text{if $n$ is odd} \\
n-1 &  \text{if $n$ is even} \\
\end{array}
\right.
\quad
\beta(n)=  \left\{ 
\begin{array}{ll}
n+3 & \text{if $n$ is even} \\
n-3 & \text{if $n$ is odd and $n > 3$} \\
1&\text{if $n=3$}\\
3&\text{if $n=1$}.
\end{array}
\right.
\] 
An easy computation then shows that $\alpha$ yields $x \prec y$, whereas $\beta$ yields $y^{\prime} \prec x$, since we have noticed that $y \preceq y^{\prime}$, we thus obtain a contradiction.  
Note that the construction actually shows that even WE  and M are inconsistent on $Y^\N$.}
\end{example}

Both examples involve utility domains $Y$ having an infinite descending chain. 
Indeed if we consider utility domains $Y$ without infinite descending chains, possibility results can be established, as it is shown in Lemma \ref{L3}. 
Note that the same impossibility results occur when we replace GE (and WE) with GPD (and WPD).

\section{Representation of social welfare relations satisfying generalized equity} \label{s4}

Having remarked some restrictions and sufficient conditions to be imposed in order to have SWRs  satisfying combinations of GE, AN and M, we can now move to investigate results about representation via SWFs. 
First of all, we present two possibility results for real-valued representation satisfying GE.

\begin{proposition}\label{P1}
Let $Y\equiv \{a, b, c, d, e\}\subset\RR$ be such that $a<b<c<d<e$ and let $X\equiv Y^{\N}$.
For any sequence $x\in X$, define:%
\begin{equation*}
N(x) = \{n:n\in \N\;\text{and}\;x_{n}=a\},\;\text{and}\;M(x) = \{m:m\in \N\;\text{and}\;x_{m}=b\}.
\end{equation*}%
Define the social welfare function $W:X\rightarrow \RR$ by 
\begin{equation}\label{P1E1}
W(x)=%
\begin{cases}
-\underset{n\in N(x)}{\sum}\frac{1}{2^{n}} - \underset{m\in M(x)}{\sum}\frac{1}{3^{n}} & \;\text{if}\;N(x)\;\text{or}\;M(x)\;\text{is non-empty} \\ 
\overset{\infty}{\underset{n=1}{\sum}}\frac{x_{n}}{2^{n}} & \;\text{otherwise}.%
\end{cases}
\end{equation}%
Then $W(x)$ satisfies GE and M.%
\end{proposition}

\begin{proposition}\label{P2}
Let $Y=\{a, b, c, d, e, f, g\}\subset\RR$ where $a<b<c<d<e<f<g$.
For every $x\in X$, define,
\begin{equation*}
N(x)=\{n: n\in \N, x_{n}=a\; \text{or}\; x_n=g\}\;\text{and}\; M(x)=\{n: n\in \N, x_{n}=b\; \text{or}\; x_n=f\}.
\end{equation*}
Define the social welfare function $W:X\rightarrow \RR$ by:
\begin{equation} \label{P2E1}
W(x)=
\begin{cases}
-\underset{n\in N(x)}{\sum}\frac{1}{2^{n}} - \underset{n\in M(x)}{\sum}\frac{1}{3^{n}}&\; \text{if}\; N(x)\; \text {or}\; M(x) \;\text{is non-empty},\\
0& \; \text{otherwise}.
\end{cases}
\end{equation}
Then $W$ satisfies GE.
\end{proposition}

The representation in Proposition \ref{P1} shows that the monotone SWF reported in \citet[Proposition 1]{dubey2014} could easily be shown to satisfy GE, while the representation in Proposition \ref{P2} is based on SWF satisfying SE for $Y$ consisting of seven distinct elements from an unpublished version of \citet{dubey2014}.
These two results yield the following two natural questions:
\begin{question} \label{Q1}
\emph{What is the role of AN? 
The SWFs defined in Proposition \ref{P1} and \ref{P2} do not satisfy AN since they are discounted sums. 
Hence, it is natural to ask whether there exists a more general SWF satisfying AN as well, or if there are some deeper structural limits in combining GE/GPD with AN.}
\end{question}
\begin{question} \label{Q2}
\emph{In both proofs, the cardinality of the utility domain $Y$ seems to be crucial, and indeed the proofs do not work when we add more elements to the utility domains. 
Is it possible to improve either of the two constructions or are there any structural limits?}
\end{question}

In the next sections we tackle and provide answers to these questions. 
We will deal with GE, IE and WE, but in most cases analogous results can be proven with the same arguments for their corresponding generalized Pigou-Dalton principles (GPD, IPD, WPD) as well.   
This is due to an \emph{invariance} result, which is the content of Lemma \ref{L1} below. 
In words, the lemma states that any representability result is invariant with respect to monotone transformations of the utility domain $Y$.

\begin{lemma}\label{L1}
Let $Y$ be a non-empty subset of $\RR$, $X\equiv Y^{\N}$.
Let $\preccurlyeq$ be a social welfare relation and $W:X\rightarrow \RR$ be a social welfare function satisfying generalized equity and anonymity.
Suppose $f$ is a monotone (increasing or decreasing) function from $I$ to $Y$, where $I$ is a non-empty subset of $\RR$.
Then, there exist a social welfare relation  $\preccurlyeq_{J}$ and a social welfare function  $V:J\rightarrow \RR$ satisfying generalized equity and anonymity, where $J=I^{\N}$.
\end{lemma}

It is important to note that the invariance result holds for only monotone increasing transformations when we combine generalized equity with monotonicity axiom.

\subsection{Infinite equity and anonymity} 

In this section we examine Question \ref{Q1}. 
It turns out that the non-existence of a SWF satisfying our equity principles together with AN occurs for all non-trivial domains even when GE is weakened to IE. 
This demonstrates inherent conflict among real-valued representation, IE and AN.

\begin{theorem} \label{T1}
Let $X=Y^{\N}$. 
Then there does not exist any social welfare function satisfying infinite equity and anonymity on $X$ if $|Y|\geq 4$.
\end{theorem}

\begin{proof}
We establish the claim by contradiction.
Let there be a SWF $W: X\rightarrow \RR$ satisfying IE and AN.
Let $q_1$, $q_2$, $\cdots$ be an enumeration of all the rational numbers in $(0, 1)$.
We keep this sequence fixed for the entire proof.
Let $r\in (0, 1)$.
We construct a binary sequence $x(r)$ relying on the enumeration of rational numbers as follows.
Define $L(r)$ as follows.
Let $l_1= \min \left\{ n\in \N: q_n \in (0, r)\right\}$.
Having defined $l_1$, $l_2$, $\cdots$, $l_k$, we get 
\[
l_{k+1} = \min \left\{ n\in \N \setminus \{l_1, l_2, \cdots, l_k\}: q_n \in (0, r)\right\}.
\]
Thus 
\[
L(r) := \{l_1(r), l_2(r), \cdots, l_k(r), \cdots\}, \; \text{with}\; l_1(r)<l_2(r)<\cdots.
\]
Further, $U(r)$ denotes the set $\N\setminus L(r)$ and 
\[
U(r) := \{u_1(r), u_2(r), \cdots, u_k(r), \cdots\}, \; \text{with}\; u_1(r)<u_2(r)< \cdots <\cdots.
\]
\noindent The utility streams  $\langle x(r)\rangle$, and $\langle y(r)\rangle$ are defined as follows.%
\begin{equation}\label{03}
x_t(r) = \left\{ 
\begin{array}{ll}
b & \text{for } t =2n-1, n \in L(r) \\
c & \text{for } t =2n, n \in L(r) \\
a & \text{for } t =2n-1,  n \in U(r) \\
d & \text{for } t =2n, n \in U(r).
\end{array}
\right.
\end{equation}
Next, we choose a strictly decreasing sequence of rational numbers $\langle q_{\textbf{n}_k}(r)\rangle\in (r, 1)$ from the (fixed) enumeration which is convergent to $r$, and consider the sub-sequence $\left\{\textbf{n}_k: k\in \N\right\}$.
We change elements of the sequence $\langle x(r)\rangle$ at the coordinates $(2\textbf{n}_k-1, 2\textbf{n}_k)$, from $a$ and $d$  to $b$ and $c$, respectively so as to obtain the utility stream $\langle y(r)\rangle$:
\begin{equation}\label{04}
y_t(r) = \left\{ 
\begin{array}{ll}
b & \text{for } t =2\textbf{n}_k-1, k\in\N \\
c & \text{for } t =2\textbf{n}_k, k\in\N\\
x_t(r) & \text{otherwise.}
\end{array}
\right.
\end{equation}
We claim that $x(r) \prec y(r)$. 
Note that for each $k\in \N$, 
\[
x_{2\textbf{n}_k-1}(r) = a < b = y_{2\textbf{n}_k-1}(r) < c = y_{2\textbf{n}_k}(r) < d = x_{2\textbf{n}_k}(r).
\]
Hence, by IE we have $x(r)\prec y(r)$, and therefore, 
\begin{equation}\label{1PE3}
W(x(r))<W(y(r)).
\end{equation}
Next, we pick $s\in (r, 1)$ and define the utility stream $\langle x(s)\rangle$ using the same construction as in (\ref{03}).
We will show that $y(r) \prec x(s)$. 
Note that for all but finitely many $\textbf{n}_k$, we have $q_{\textbf{n}_k}<s$.
Pick $\textbf{n}_1 < \textbf{n}_2 < \cdots< \textbf{n}_K$, ($K$ finite) such that $q_{\textbf{n}_1}\geq s$, $\cdots$, $q_{\textbf{n}_K}\geq s$.
We have for all $k\leq K$, 
\[
x_{2\textbf{n}_k-1}(s) = a < b = y_{2\textbf{n}_k-1}(r) < c = y_{2\textbf{n}_k}(r) < d = x_{2\textbf{n}_k}(s).
\]
However, these are only finitely many such pairs of coordinates.
There exist infinitely many $l_m(s) \in \N \setminus \{\textbf{n}_1, \textbf{n}_2, \cdots, \textbf{n}_K\}$, with $l_m(s)>\textbf{n}_K$, such that for all $k \geq K$, $l_{m}(s) \neq \textbf{n}_k$ and $q_{l_m(s)}\in [r, q_{\textbf{n}_K})$.
For each of these $l_m(s)$,
\[
y_{2l_m(s)-1}(r) = a < b = x_{2l_m(s)-1}(s) < c = x_{2l_m(s)}(s) < d = y_{2l_m(s)}(r).
\]
We interchange the pairs of coordinates  $(2\textbf{n}_1-1, 2\textbf{n}_1)$, $\cdots$, $(2\textbf{n}_K-1, 2\textbf{n}_K)$ of $\langle y(r)\rangle$ (having value $b$ and $c$ at odd and even locations) with equally many elements (having $a$ and $d$ at odd and even locations) from the coordinates $(2l_m(s)-1, 2l_m(s))$ so as to obtain utility stream $\langle y^{\prime}\rangle$. 
Applying AN, $y(r)\sim y^{\prime}$ and therefore,
\begin{equation}\label{1PE4}
W(y(r)) = W(y^{\prime}).
\end{equation}
Comparing $\langle y^{\prime}\rangle$ to $\langle x(s)\rangle$, observe that for all $n\in U(s)$, $y^{\prime}_{2n-1} = x_{2n-1}(s)=a$, $y^{\prime}_{2n} = x_{2n}(s)=d$; for all $n\in L(r)$, $y^{\prime}_{2n-1} = x_{2n-1}(s)=b$, $y^{\prime}_{2n} = x_{2n}(s)=c$; for all $\textbf{n}_k$, with $k>K$ $y^{\prime}_{2n-1} = x_{2n-1}(s)=b$, $y^{\prime}_{2n} = x_{2n}(s)=c$.
Further, 
\[
y^{\prime}_{2v_{k^{\prime}}-1} = a < b = x_{2v_{k^{\prime}}-1}(s) < c = x_{2v_{k^{\prime}}}(s) < d = y^{\prime}_{2v_{k^{\prime}}}
\]
for every $v_{k^{\prime}}\in (U(r)\cap L(s)) \setminus \{\textbf{n}_K, \textbf{n}_{K+1}, \cdots\}$.
Since  there are infinitely many $v_{k^{\prime}}$, by IE $y^{\prime}\prec x(s)$, and consequently,
\begin{equation}\label{1PE5}
W(y^{\prime}) <W(x(s)).
\end{equation}
From (\ref{1PE4}) and (\ref{1PE5}), we get
\begin{equation}\label{1PE6}
W(y(r)) <W(x(s)).
\end{equation}

Therefore, observe that, given $r<s$, $(W(x(r)), W(y(r))$ and $(W(x(s)), W(y(s))$ are non-empty and disjoint open intervals.    
Hence, by density of $\mathbb{Q}$ in $\RR$ we conclude that we have found a one-to-one mapping from $(0, 1)$ to $\mathbb{Q}$, which is impossible as the latter set is countable.
\end{proof}

\subsection{Infinite equity on larger utility domains}

In this section, we answer Question \ref{Q2}.
The following result shows that Proposition \ref{P1} cannot be improved in order to require $Y$ consisting of six utility values, or more.

\begin{theorem} \label{T2}
The following are equivalent:
\begin{enumerate}
\item[(i)] There exists a social welfare function $W: Y^\N \rightarrow \mathbb{R}$ satisfying GE and M.
\item[(ii)] There exists a social welfare function $W: Y^\N \rightarrow \mathbb{R}$ satisfying IE and M.
\item[(iii)] $Y$ contains at most 5 elements.
\end{enumerate}
\end{theorem}

Next result shows that we cannot improve Proposition \ref{P2} in order to require $Y$ consisting of eight utility values or more. 

\begin{theorem}\label{T3}
The following are equivalent:
\begin{enumerate}
\item[(i)] There exists a social welfare function $W: Y^\N \rightarrow \mathbb{R}$ satisfying GE.
\item[(ii)] There exists a social welfare function $W: Y^\N \rightarrow \mathbb{R}$ satisfying IE.
\item[(iii)] $Y$ contains at most 7 elements.
\end{enumerate}
\end{theorem}

\begin{remark} \label{R2}
\emph{It is interesting to notice distinct characteristics of the pairing functions involved in the proof of Theorem \ref{T1} from those involved in the proofs of Theorems \ref{T2} and \ref{T3} (see appendix). 
The pairing functions used in Theorem \ref{T1} have very simple combinatorial properties, as they link two consecutive  generations of rich and poor. 
On the contrary, the pairing functions $\alpha$ ($\beta$) used in Theorems \ref{T2} and \ref{T3} are not necessarily \emph{bounded} in the following sense: for every $N \in \N$ there exists $n \in \N$ such that $|\alpha(n) -n| > N$ ($|\beta(n) -n| > N$).}
\end{remark}

\subsection{Social welfare relations satisfying weak equity}
Having observed the restrictive nature of domain $Y$ for the existence of SWF satisfying IE, it makes sense to weaken the equity concept even further to check if there exists a SWF on a larger domain. 
We consider the weakest notion, i.e., WE here.
Following lemma shows that there exists a SWR satisfying GE, AN and M when $Y(<)$ is a well-ordered set of real numbers.%
\footnote{In case $Y$ is not a well-ordered set, Example \ref{Ex2} shows the non-existence of SWR satisfying $\we$ and M.} 
\begin{lemma}\label{L3}
Let $Y \subseteq [0,1]$ be well-ordered. Then there exists a social welfare relation satisfying generalized equity, anonymity and monotonicity on $Y^\N$. 
Moreover, using the axiom of choice, one can obtain a social welfare order satisfying generalized equity, anonymity and monotonicity on $Y^\N$.
\end{lemma}
Next theorem provides a precise characterization of representation for WE, AN and M.

\begin{theorem} \label{T4}
The following are equivalent:
\begin{enumerate}
\item[(i)]  There exists a social welfare function $W: Y^\N \rightarrow \mathbb{R}$ satisfying WE, M and AN.
\item[(ii)]  There exists a social welfare function $W: Y^\N \rightarrow \mathbb{R}$ satisfying WE, M.
\item[(iii)]  $Y$ is well-ordered.
\end{enumerate}
\end{theorem}

\begin{remark} \label{R3}
\emph{Note that $(ii)$ reveals that adding AN in this context does not alter the condition for the existence of SWF satisfying WE and M. 
In fact, Theorem \ref{T4} (specifically the equivalence $(i) \Leftrightarrow (ii)$) states that the existence of real-valued representation for WE, M and AN is equivalent to the existence of real-valued representation for WE and M (without AN). 
This is in sharp contrast with the situation occurring for IE; indeed Theorem \ref{T2} shows that real-valued representation for IE and M is possible if $Y$ has at most five elements, whereas Theorem \ref{T1} shows that IE and AN together never admit a real-valued representation (for non-trivial utility domains). 
So when dealing with IE, adding AN strongly effects the structure of the utility domain $Y$ for which representation is possible.}
\end{remark}

Now consider the situation when we drop M.
In Example \ref{Ex1} we have shown that there does not exist a SWR satisfying WE when $Y$ contains a subset isomorphic to the set of integers. 
Next theorem provides a precise characterization of representation for WE and AN.

\begin{theorem} \label{T5}
The following are equivalent:
\begin{enumerate}
\item[(i)]  There exists a social welfare function $W: Y^\N \rightarrow \mathbb{R}$ satisfying WE and AN.
\item[(ii)]  There exists a social welfare function $W: Y^\N \rightarrow \mathbb{R}$ satisfying WE.
\item[(iii)]  $Y\subset [0, 1]$ does not contain a subset order isomorphic to the set of integers $\mathbb{Z}$.
\end{enumerate}
\end{theorem}

Note that the same reasoning on the role of AN as highlighted in Remark \ref{R3} also applies for Theorem \ref{T5}.

\section{Conclusions}\label{s5}
In this paper we have proposed an intuitive  generalization of strong equity and Pigou-Dalton transfer principle.
We have investigated the conditions for the existence of social welfare relations satisfying the new equity condition (by itself or in conjunction with anonymity and/or monotonicity axioms) and their real-valued representation.
The results indicate that the order properties of the one period utility domain $Y$ determine the conditions for the existence of social welfare relations satisfying GE and/or M.
In other words, cardinality or measure (size) of $Y$ is not the relevant property for the existence of the equitable social welfare relations.
However, when one investigates the real-valued representations, cardinality of $Y$ does play a crucial role in some cases.

Tables \ref{Ta1} and \ref{Ta2} below summarize our main results.
We abbreviate consequentialist equity by CE, procedural equity by PE and efficiency by E in the tables below.

\begin{table}[ht]
\centering
\begin{tabular}{l|c |c| c |c |c }
\hline\hline 
CE&PE&E&$|Y|$& Representation & Result \\ [0.25ex] 
\hline\hline
GE&-&M &$\leq 5$&Yes& Proposition \ref{P1}\\\hline
GE&-&-&$\leq 7$&Yes&Proposition \ref{P2}\\ \hline
IE&AN&-&$\geq 4$& No&Theorem \ref{T1}\\ \hline
IE&-&M &$\geq 6$&No&Theorem \ref{T2}\\\hline
IE&-&-&$\geq 8$&No&Theorem \ref{T3}\\ \hline
\end{tabular}
\caption{Possibility vs impossibility results for GE and IE}\label{Ta1}
\end{table}

\begin{table}[ht]
\centering
\begin{tabular}{l|l | c | c |l |l }
\hline\hline 
CE&PE&E&$Y(<)$& SWR & Result \\ [0.25ex] 
\hline\hline
WE&AN&M&well-ordered&admits SWF& Theorem \ref{T4}\\\hline
WE&-&M&not well-ordered&does not exist& Example \ref{Ex2}\\\hline
WE&AN&-&no subset of order type $\sigma$&admits SWF& Theorem \ref{T4}\\\hline
WE&-&-&subset of order type $\sigma$&does not exist& Example \ref{Ex1}\\\hline
\end{tabular}
\caption{Possibility vs impossibility results for WE}\label{Ta2}
\end{table}


The results summarized in Tables \ref{Ta1} and \ref{Ta2} describe the exact restrictions on the utility domain $Y$ for representation of equitable social welfare relations.
We have expressed explicit formulae of the social welfare functions whenever they exist.
Therefore these social welfare functions could be useful in policy formulations.

In future research, we consider extending the study in two distinct directions.
Careful scrutiny of the proofs of Theorems \ref{T2} and \ref{T3} reveals the pairing functions let the distance (number of generations) between the rich and the poor agents $|n-\alpha(n)|$ to be arbitrary; in fact, $|n-\alpha(n)|$ diverges to infinity.
It will be interesting to explore if these impossibility results persist even when $|n-\alpha(n)|$ is not allowed to exceed some large but finite integer (such as in Theorem \ref{T1}), or more generally when we require more structural characteristics of the pairing functions.
Such restrictions possess intuitive appeal in the sense that it makes the redistribution relatable to the individual generations.
It would also be worthwhile to explore variants of generalized equity weaker than IE which are still stronger than WE (for instance asymptotic equity%
\footnote{Asymptotic equity is defined as follows: Given $x$, $y \in X$ if there is $\alpha \in \Pi$ such that $d(\dom(\alpha))>0$ and for every $i \in \dom(\alpha)$ one has either $y_i < x_i < x_{\alpha(i)} < y_{\alpha(i)}$ or $y_{\alpha(i)} < x_{\alpha(i)} < x_i < y_i$ and for all $k \notin \dom(\alpha)$ one has $x_k = y_k$, then $y \prec x$.
Here $d(A)$ is the asymptotic density of $A\subset \N$.}) and their role in evaluation of infinite utility streams.
We intend to report results on this research in near future.

The second line of research deals with the social welfare orders when they exist but do not admit real-valued representations.
There exist preference relations which are complete (i.e., preference orders) with no social welfare function and yet are endowed with explicit description.
The lexicographic preference is the most commonly cited example of such preference orders.
Given explicit formulation of the order, it is still considered to be useful in policy making.
Therefore, it is a fruitful exercise to explore if the equitable social welfare orders which have been shown to exist in Lemma \ref{L3} but fail to admit SWF (considered in Theorems \ref{T1}, \ref{T2} and \ref{T3}) are equipped with explicit description.
Outcomes of this investigation would add to the rich body of results on constructive nature of equitable social welfare orders.

\section{Appendix}

\begin{proof}[\large \textbf{Proof of Lemma \ref{L1}}]
We treat two cases: (i) $f$ is increasing, and (ii) $f$ is decreasing.
\begin{enumerate}
\item[(i)]{Let $f$ be an increasing function. 
Define $\preccurlyeq_{J}$ as 
\begin{equation}\label{0L1}
\left\langle x_1, x_2, \cdots\right\rangle \preccurlyeq_{J} \left\langle y_1, y_2,\cdots\right\rangle \;\ifif \left\langle f(x_1), f(x_2), \cdots\right\rangle \preccurlyeq \left\langle f(y_1), f(y_2),\cdots\right\rangle,
\end{equation}%
and $V:J\rightarrow \RR$ by:%
\begin{equation}\label{0L2}
V\left(z_{1}, z_{2},\cdots\right) = W\left(f(z_{1}), f(z_{2}), \cdots\right)  
\end{equation}%
Then, $\preccurlyeq_{J}$ and $V$ are well-defined, since $f$ maps $I$ into $Y$.
To verify AN, let $z$, $z^{\prime}\in J$, with $z_{r}^{\prime} = z_{s}$, $z_{s}^{\prime}=z_{r}$, and $z_{i}^{\prime}=z_{i}$ for all $i\neq r, s$. 
Without loss of generality, assume $r<s$.
Then,
\begin{align}\label{0L3}
\left(z_{1}^{\prime}, z_{2}^{\prime}, \cdots\right) &= \left(f(z_{1}^{\prime}), f(z_{2}^{\prime}), \cdots, f(z_{r}^{\prime}), \cdots, f(z_{s}^{\prime}), \cdots\right)\notag\\
&\sim \left(f(z_{1}^{\prime}), f(z_{2}^{\prime}), \cdots, f(z_{s}^{\prime}), \cdots, f(z_{r}^{\prime}),\cdots\right)\\
&=\left(f(z_{1}), f(z_{2}), \cdots, f(z_{r}), \cdots, f(z_{s}), \cdots\right) = \left(z_{1}, z_{2}, \cdots\right),\notag
\end{align}
the second line of (\ref{0L3}) following from the fact that $\preccurlyeq$ satisfies AN on $X$.
Hence, $\left(z_{1}^{\prime}, z_{2}^{\prime}, \cdots\right) \preccurlyeq_J \left(z_{1}, z_{2}, \cdots\right)$.
Note that the fact that $f$ is increasing is nowhere used in this demonstration.
Identical argument works for the SWF.

To check that the SWF $V$ satisfies GE, let $z$, $z^{\prime}\in J$ with $\alpha \in \Pi$ such that for every $i \in \dom(\alpha)$ one has
\[
\text{either}\; z_i < z^{\prime}_i < z^{\prime}_{\alpha(i)} < z_{\alpha(i)}\; \text{or}\; z_{\alpha(i)} < z^{\prime}_{\alpha(i)} < z^{\prime}_i < z_i.
\]
We consider the case $z_i < z^{\prime}_i < z^{\prime}_{\alpha(i)} < z_{\alpha(i)}$ (proof in remaining case works on similar line).
We have
\begin{equation}\label{0L4}
f(z_{i})< f(z_{i}^{\prime}) < f(z^{\prime}_{\alpha(i)}) < f(z_{\alpha(i)})\;\;\text{for each} \; i\in \alpha, 
\end{equation}%
since $f$ is increasing. 
Consequently,%
\begin{align*}
V(z_{1}^{\prime}, z_{2}^{\prime}, \cdots)=W(f(z_{1}^{\prime}), f(z_{2}^{\prime}), \cdots)  > W(f(z_{1}), f(z_{2}), \cdots) = V(z_{1}, z_{2}, \cdots) 
\end{align*}%
where the inequality follows from the facts that $W$ satisfies GE on $X$, $f(z_{i})\in Y$,  $f(z_{i}^{\prime})\in Y$ for all $i\in \N$ and (\ref{0L4}) holds for all $i\in \alpha$.}

\item[(ii)]{Let $f$ be a decreasing function. 
We consider $\preccurlyeq_{J}$  and $V:J\rightarrow \RR$  as defined in (\ref{0L1}) and (\ref{0L2}) respectively.
One can check that AN is satisfied by following the steps used in (i) above. 
To check that the SWF $V$ satisfies GE, let $z$, $z^{\prime}\in J$ with $\alpha \in \Pi$ such that for every $i \in \dom(\alpha)$ one has
\[
\text{either}\; z_i < z^{\prime}_i < z^{\prime}_{\alpha(i)} < z_{\alpha(i)}\; \text{or}\; z_{\alpha(i)} < z^{\prime}_{\alpha(i)} < z^{\prime}_i < z_i.
\]
We consider the case $z_i < z^{\prime}_i < z^{\prime}_{\alpha(i)} < z_{\alpha(i)}$ (proof in remaining case works on similar line).
We have
\begin{equation}\label{0L5}
f(z_{\alpha(i)}) < f(z^{\prime}_{\alpha(i)}) < f(z_{i}^{\prime}) < f(z_{i}) \;\;\text{for each} \; i\in \alpha, 
\end{equation}%
since $f$ is decreasing. 
Consequently,%
\begin{align*}
V(z_{1}^{\prime}, z_{2}^{\prime}, \cdots)= W(f(z_{1}^{\prime}), f(z_{2}^{\prime}), \cdots) > W(f(z_{1}), f(z_{2}), \cdots) = V(z_{1}, z_{2}, \cdots)
\end{align*}%
where the inequality follows from the facts that $W$ satisfies GE on $X$, $f(z_{i})\in Y$,  $f(z_{i}^{\prime})\in Y$ for all $i\in \N$ and (\ref{0L5}) holds for all $i\in \alpha$.
Similar conclusions can be drawn about the $\preccurlyeq_{J}$ following the argument above.}
\end{enumerate}
\end{proof}

\begin{proof}[\large \textbf{Proof of Lemma \ref{L3}}]
Let $Y \subseteq [0,1]$ be well-ordered and let $\overline{x}[n]$ be the non-decreasing reordering of $x[n]$, which means, more precisely, 
\[
\overline{x}[n] := \langle \overline{x}[n]_1, \overline{x}[n]_2, \cdots, \overline{x}[n]_n \rangle
\] 
is recursively defined on $k \leq n$, as follows:
\[
\begin{split}
\overline{x}[n]_1 &:= \min\{x_i: i \leq n \} \\
\overline{x}[n]_{k} &:= \min\{x_i: i \leq n \land \forall j <k (x_i \geq \overline{x}[n]_j).
\end{split}
\]
Then consider the lexicographic order $\leq^n_{\text{lex}}$ on $Y^n$. 
Let $\mathcal{F}$ be any filter on $\N$ containing all co-finite subsets of $\N$.
Define the relation $\leq_{\mathcal{F}}$ on $Y^\N$ as follows: for every $x$, $y \in Y^\N$
\[
x \leq_\mathcal{F} y \quad \ifif \quad \{n \in \N: \overline{x}[n] \leq^n_\text{lex} \overline{y}[n]  \} \in \mathcal{F}.
\]
We need to check that $\leq_\mathcal{F}$ satisfies GE, AN and M.
AN and M are trivial. 
To show that SWR satisfies GE pick $x$, $y \in Y^\N$ such that there exists $\alpha \in \Pi$ such that for all $n \in \dom(\alpha)$,
\begin{equation} \label{eq1}
\text{either} \quad  x_n < y_n < y_{\alpha(n)} < x_{\alpha(n)} \quad \text{or} \quad x_{\alpha(n)} < y_{\alpha(n)} < y_n < x_n. 
\end{equation}
We show that there exists $\mathbf{k} \in \N$ such that for all $m \geq k$ one has $\overline{x}[m] \leq^m_\text{lex} \overline{y}[m]$. 
Let $h^x:= \min\{ x_n: n \in \dom(\alpha) \}$ (and analogously define $h^y$); then put $\textbf{k}:=\min\{n \in \N: x_n=h^x\}$. 
We show that for all $m \geq \textbf{k}$ we get $\overline{x}[m] <^m_{\text{lex}} \overline{y}[m]$ in two steps.
\begin{enumerate}[(a)]
\item First, pick $\textbf{l} \leq \textbf{k}$ such that $\overline{x}[\textbf{k}]_{\textbf{l}} = h^x < \overline{y}[\textbf{k}]_{\textbf{l}}$.
Since for every $i < {\textbf{l}}$ one has $\overline{x}[\textbf{k}]_i= \overline{y}[\textbf{k}]_i$, $\overline{x}[\textbf{k}] <^{\textbf{k}}_\text{lex} \overline{y}[\textbf{k}]$.
\item We next show for all $m \geq \textbf{k}$, $\overline{x}[m] <^m_{\text{lex}} \overline{y}[m]$ by induction.
The details for the step $m=\textbf{k}+1$ illustrate the general $m \geq \textbf{k}$. 
We can distinguish two cases:
\begin{itemize}
\item $m \in \dom(\alpha)$: $x_m \geq h^x$ and $y_m > h^x$ by definition of $h^x$.
For all $i < \textbf{l}$, $\overline{x}[m]_i= \overline{x}[\textbf{k}]_i = \overline{y}[\textbf{k}]_i= \overline{y}[m]_i$, and $\overline{x}[m]_{\textbf{l}}= h^x < \overline{y}[m]_{\textbf{l}}$. 
Hence, $\overline{x}[m] <^m_\text{lex} \overline{y}[m]$.
\item $m \notin \dom(\alpha)$: Two possibilities need to be considered.
\begin{enumerate}[(i)]
\item{$x_m \geq h^x$:  Then $\overline{x}[m]_{\textbf{l}}= \overline{x}[\textbf{k}]_{\textbf{l}}< \overline{y}[\textbf{k}]_{\textbf{l}} = \overline{y}[m]_{\textbf{l}}$ and for all $i < {\textbf{l}}$, $\overline{x}[m]_i = \overline{y}[m]_i$.
Hence, $\overline{x}[m] <^m_\text{lex} \overline{y}[m]$.}
\item{$x_m < h^x$: Since $x_m=y_m$, we have for all $i \leq \textbf{l}$, $\overline{x}[m]_i = \overline{y}[m]_i$ and $\overline{x}[m]_{\textbf{l}+1}= \overline{x}[\textbf{k}]_{\textbf{l}}= h^x < \overline{y}[\textbf{k}]_{\textbf{l}} = \overline{y}[m]_{\textbf{l}+1}$.
Hence, $\overline{x}[m] <^m_\text{lex} \overline{y}[m]$.}
\end{enumerate}
\end{itemize}
\end{enumerate}

We have thus shown that $\textbf{k}$ satisfies the required property, and therefore $x <_{\mathcal{F}} y$, since $\mathcal{F}$ contains all co-finite subsets of $\N$.
Finally, notice that if $\mathcal{U} \supseteq \mathcal{F}$ is an ultrafilter then $\leq_{\mathcal{U}}$ is even total and thus we have obtained a SWO. 
\end{proof}

\begin{proof}[\large \textbf{Proof of Proposition \ref{P1}}]
\begin{enumerate}
\item[(a)]{GE: Pick $x$, $y \in Y^\N$ such that there exists $\alpha \in \Pi$ such that for all $n \in \dom(\alpha)$, either $x_n < y_n < y_{\alpha(n)} < x_{\alpha(n)}$ or $x_{\alpha(n)} < y_{\alpha(n)} < y_n < x_n$.
Consider the case $x_n < y_n < y_{\alpha(n)} < x_{\alpha(n)}$.
Note that either $x_n=a$ or $x_n =b$, i.e., either $n\in N(x)$ or $n\in M(x)$.
For $n\in N(x)\cap \alpha$, let $N(xb) := \{n: y_{n}=b\}$ and $N(xc) := \{n:y_{n}=c\}$.
Then
\[
W(y) -W(x) = \underset{n\in N(xb)}{\sum}\left(\frac{1}{2^{n}} - \frac{1}{3^{n}}\right) + \underset{n\in N(xc)}{\sum}\frac{1}{2^{n}} + \underset{n\in M(x)\cap \alpha}{\sum}\frac{1}{3^{n}}>0.
\]}
\item[(b)]{M: Let $y\geq x$. 
Then $N(y)\subset N(x)$ and $M(y)\subset M(x)$ in case $N(x)$ or $M(x)$ is non-empty.
Then 
\[
W(y)-W(x) = \underset{n\in N(x)\setminus N(y)}{\sum}\frac{1}{2^{n}} +\underset{n\in M(x)\setminus M(y)}{\sum} \frac{1}{3^{n}}\geq 0.
\]
If $N(x)$ and $M(x)$ are non-empty, then 
\[
W(y) =  \overset{\infty}{\underset{n=1}{\sum}}\frac{y_{n}}{2^{n}}\geq \overset{\infty}{\underset{n=1}{\sum}}\frac{x_{n}}{2^{n}}=W(x).
\]}
\end{enumerate}
Moreover, it can be checked that $W$, defined by (\ref{P1E1}), also satisfies weak Pareto; i.e.,  if $x$, $y\in X$, and $x\gg y$, then $W(x)>W(y)$.
\end{proof}

\begin{proof}[\large \textbf{Proof of Proposition \ref{P2}}]
Pick $x$, $y \in Y^\N$ such that there exists $\alpha \in \Pi$ such that for all $n \in \dom(\alpha)$, either $x_n < y_n < y_{\alpha(n)} < x_{\alpha(n)}$ or $x_{\alpha(n)} < y_{\alpha(n)} < y_n < x_n$.
Consider the case $x_n < y_n < y_{\alpha(n)} < x_{\alpha(n)}$.
We have to show that $W(y)>W(x)$.
Note that $y_{\alpha(n)}\in \{c, d, e, f\}$.
So, the contribution of $y_{\alpha(n)}$ to $W(y)$ (call it $U(y)$) is either $-\frac{1}{3^{\alpha(n)}}$ or $0$.
If it is $-\frac{1}{3^{\alpha(n)}}$, then $x_{\alpha(n)}=g$ and contribution of $x_{\alpha(n)}$ to $W(x)$ (i.e., $U(x)$) is $-\frac{1}{2^{\alpha(n)}}<-\frac{1}{3^{\alpha(n)}}$; and if $U(y)=0$, then $U(x)\leq 0=U(y)$.
Thus, in either case,
\begin{equation}\label{P2E2}
U(x)\leq U(y).  
\end{equation}
Furthermore, the inequality is strict in (\ref{P2E2}) if $y_{\alpha(n)}\in \{e, f\}$.
Note that $y_{n}\in \{b, c, d, e\}$.
So, contribution of $y_n$ to $W(y)$ (call it $V(y)$) is either $-\frac{1}{3^{n}}$ or $0$.
If it is $-\frac{1}{3^{n}}$, then $x_n=a$ and $V(x)$ is $-\frac{1}{2^{n}}<-\frac{1}{3^{n}}$; and if $V(y)=0$, then $V(x)\leq 0=V(y)$.
Thus, in either case, 
\begin{equation}\label{P2E3}
V(x)\leq V(y) 
\end{equation}%
Furthermore, the inequality is strict in (\ref{P2E3}) if $y_{n}\in \{b, c\}$.
If neither (\ref{P2E2}) nor (\ref{P2E3}) holds with a strict inequality, then we must have $y_{\alpha(n)}\in \{c, d\}$ and $y_{n}\in \{d, e\}$.
But then $y_{n}\geq d \geq y_{\alpha(n)}$, a contradiction. 
Thus at least one of the inequalities in (\ref{P2E2}) and (\ref{P2E3}) is strict for each $n\in\alpha$, and so $W(x)<W(y)$.
\end{proof}

\begin{proof}[\large \textbf{Proof of Theorem \ref{T2}}]

We introduce the following partition of $\N$ which will be useful in the proof of Theorems \ref{T2} and \ref{T3}.
Let $q_1$, $q_2, \cdots$ be an enumeration of all the rational numbers in $(0, 1)$.
Given any $r\in (0, 1)$, $L(r):= \left\{l_1, l_2, \cdots, l_k, \cdots\right\}$ and $U(r):= \{u_1(r), u_2(r), \cdots, u_k(r), \cdots\}$ are defined as above in the proof of Theorem \ref{T1}.
Then let
\[
\begin{split}
\textbf{L}(r) &:= \{2(l_1!)+1, 2(l_1!)+2, 2(l_2!)+1, 2(l_2!)+2, \dots\}:= \{ 2(l_k!)+1, 2(l_k!)+2: k \in \N \}  \\
\textbf{U}(r) &:= \{2(u_1!)+1, 2(u_1!)+2, 2(u_2!)+1, 2(u_2!)+2, \dots\}  := \{ 2(u_k!)+1, 2(u_k!)+2: k \in \N \} 
\end{split}
\]
and put $\textbf{LU}(r):= \textbf{L}(r)\cup \textbf{U}(r)$. 
Given $s,r \in (0,1)$ with $s \neq r$, note that $\mathbf{L}(r) \neq \mathbf{L}(s)$ and $\mathbf{U}(r) \neq \mathbf{U}(s)$, but $\mathbf{LU}(r) = \mathbf{LU}(s)$. Let
$\mathbf{I}:=\N \setminus \mathbf{LU}(r)$; note that $\mathbf{I}$ does not depend on the choice of $r$.
Moreover $\mathbf{I}$ is clearly infinite and for every $r \in (0,1)$ it holds
\begin{equation}\label{57P01}
\mathbf{I} \cup \textbf{L}(r) \cup \textbf{U}(r) =\N.
\end{equation}

$(i) \Rightarrow (ii)$ is trivial, since GE is stronger than IE and $(iii) \Rightarrow (i)$ follows from Proposition \ref{P1}.

It is left to show $(ii) \Rightarrow (iii)$.
We argue by contradiction. 
Let $W: Y^\N \rightarrow \RR$ be a SWF satisfying IE and M, where $Y:=\{a, b, c, d, e, f\}$ with $a<b<c<d<e<f$.
Given the partition in (\ref{57P01}), the utility streams  $\langle x(r)\rangle$, $\langle y(r)\rangle$ are defined as follows.

\begin{equation}\label{5P01}
x_t(r) = \left\{ 
\begin{array}{ll}
c & \text{for}\; t \in \mathbf{I} \cap \odd \\
f & \text{for}\; t \in \mathbf{I} \cap \even \\
a & \text{for}\; t \in \textbf{U}(r) \\
b & \text{for}\; t \in \textbf{L}(r) \\
\end{array}
\right.
\end{equation}

\begin{equation}\label{5P02}
y_t(r) = \left\{ 
\begin{array}{ll}
d & \text{for}\; t \in \mathbf{I} \cap \odd \\
e & \text{for}\; t \in \mathbf{I} \cap \even \\
x_t(r) & \text{otherwise}.
\end{array}
\right.
\end{equation}

We first claim that $x(r) \prec y(r)$.
Note for each $t\in \textbf{LU}(r)$, $x_t(r) = y_t(r)$.
Each pair of terms in $\mathbf{I}$ contains $c$ and $f$ in $x(r)$ at odd and even locations respectively.
Corresponding pair of terms are assigned values $d$ and $e$ in $y(r)$ at odd and even locations respectively.
Therefore or each pair of terms in $\mathbf{I}$
\[
x_{\odd}(r) = c < d = y_{\odd}(r) < e = y_{\even}(r) < f = x_{\even}(r).
\]
Since $\mathbf{I}:=\N\setminus \textbf{LU}(r)$ is infinite, $x(r) \prec y(r)$ by IE. 
Therefore, 
\begin{equation}\label{5P03}
W(x(r))<W(y(r)).
\end{equation}

\noindent Now pick $s \in (r, 1)$ and claim that $y(r) \prec x(s)$.

In order to prove this claim, observe that:
\begin{itemize}
\item for all $t\in \mathbf{I} \cap \even$, $y_t(r) = e < f = x_t(s)$.
\item for all $t\in\textbf{U}(s)$, $y_t(r) = a = x_t(s)$; for all $t \in \textbf{L}(r)$, $y_t(r) = b = x_t(s)$;
\item since $\mathbf{I} \cap \odd$ and $\textbf{L}(s)\cap \textbf{U}(r)$ are both infinite subsets of $\N$, we can pick a pairing function $\alpha \in \Pi$ with $|\dom(\alpha)|=\infty$ such that for each $t\in \mathbf{I}\cap \odd$ we choose $\alpha(t) \in \textbf{L}(s)\cap \textbf{U}(r)$ such that for the pair of coordinates $(t, \alpha(t))$, we have 
\[
y_{\alpha(t)}(r) = a < b = x_{\alpha(t)}(s) <  c = x_t(s) < d =y_t(r).
\]
\end{itemize}
Hence, by IE and M, it follows that $y(r)\prec x(s)$, which gives
 \begin{equation}\label{5P07}
W(y(r))<W(x(s)).
\end{equation} 
Let $\{(W(x(r)), W(y(r))): r \in (0, 1)\}$. 
We have shown that $(W(x(r)), W(y(r)))$ and $(W(x(s)), W(y(s)))$ are pairwise disjoint for each $r<s$ pair, contradicting the fact that the rationals are dense in $(0, 1)$.
\end{proof}

\begin{proof}[\large \textbf{Proof of Theorem \ref{T3}}]
$(i) \Rightarrow (ii)$ is trivial, since GE is stronger than IE and $(iii) \Rightarrow (i)$ follows from Proposition \ref{P2}.

It is left to show $(ii) \Rightarrow (iii)$
The idea of the proof is similar to the one in Theorem \ref{T2}.
We argue by contradiction.
Let $W: Y^\N\rightarrow \RR$ be a SWF satisfying IE, where $Y:= \{ a,b,c,d,e,f,g,h \}$ with $a<b<c<d<e<f<g<h$. 
Given the partition in (\ref{57P01}), the utility streams  $\langle x(r)\rangle$, and $\langle y(r)\rangle$ are defined as follows.
\begin{equation}\label{7P01}
x_t(r) = \left\{ 
\begin{array}{ll}
c & \text{for}\; t \in \mathbf{I} \cap \odd \\
f & \text{for}\; t \in \mathbf{I} \cap \even \\
a & \text{for}\; t \in \textbf{U}(r)\cap \odd\\
h & \text{for}\; t \in \textbf{U}(r)\cap \even\\
b & \text{for}\; t \in \textbf{L}(r)\cap \odd\\
g & \text{for}\; t \in \textbf{L}(r)\cap \even\\
\end{array}
\right.
\end{equation}
\begin{equation}\label{7P02}
y_t(r) = \left\{ 
\begin{array}{ll}
d & \text{for}\; t \in \mathbf{I} \cap \odd \\
e & \text{for}\; t \in \mathbf{I} \cap \even \\
x_t(r) & \text{otherwise}.
\end{array}
\right.
\end{equation}
We first claim that $x(r) \prec y(r)$.
Note for each $t\in \textbf{LU}(r)$, $x_t(r) = y_t(r)$.
Each pair of terms in $\mathbf{I}$ contains $c$ and $f$ in $x(r)$ at odd and even locations respectively.
Corresponding pair of terms are assigned values $d$ and $e$ in $y(r)$ at odd and even locations respectively.
Therefore for each pair of terms in $\mathbf{I} := \N\setminus \textbf{LU}(r)$
\[
x_{\odd}(r) = c < d = y_{\odd}(r) < e = y_{\even}(r) < f = x_{\even}(r).
\]
Since $\mathbf{I}$ is infinite, $x(r) \prec y(r)$ by IE. 
Therefore, 
\begin{equation}\label{7P03}
W(x(r))<W(y(r)).
\end{equation}
Now let $s\in (r,1)$.
In order to prove $y(r) \prec x(s)$, observe that:
\begin{itemize}
\item for all $t\in\textbf{U}(s)\cap \odd$, $y_t(r) = a = x_t(s)$; and for all $t\in\textbf{U}(s)\cap \even$, $y_t(r) = h = x_t(s)$.
\item for all $t\in \textbf{L}(r)\cap \odd$, $y_t(r) = b = x_t(s)$; and for all $t\in \textbf{L}(r)\cap \even$, $y_t(r) = g = x_t(s)$.
\item for all $t\in\textbf{L}(s)\cap \textbf{U}(r)\cap\odd$, $y_t(r) = a < b = x_t(s)$ and for all $t\in \mathbf{I}\cap \odd$, $y_t(r) = d > c =x_t(s)$. Since $\textbf{L}(s)\cap \textbf{U}(r)\cap\odd$ and $\mathbf{I} \cap \odd$ are both infinite, we can then pick $\alpha \in \Pi$ with $\dom(\alpha)=\infty$ such that for every $t \in \mathbf{I} \cap \odd$,
\[
y_{\alpha(t)}(r) = a < b = x_{\alpha(t)}(s) <  c = x_t(s) < d =y_t(r).
\] 
\item  for all $t\in\textbf{L}(s)\cap \textbf{U}(r)\cap\even$, $y_t(r) = h > g = x_t(s)$ and for all $t\in \mathbf{I}\cap \even$, $y_t(r) = e < f =x_t(s)$. Since $\textbf{L}(s)\cap \textbf{U}(r)\cap\even$ and $\mathbf{I} \cap \even$ are both infinite, we can then pick $\beta \in \Pi$ with $|\dom(\beta)|=\infty$ such that for every $t \in \mathbf{I} \cap \even$,
\[
y_{\beta(t)}(r) = h > g = x_{\beta(t)}(s) >  f = x_t(s) > e =y_t(r).
\]  
\end{itemize}
Hence, by IE we have $y(r)\prec x(s)$, which gives
\begin{equation}\label{7P04}
W(y(r))<W(x(s)).
\end{equation}
As above, this contradicts the density of $\mathbb{Q}$ in $\RR$. 
\end{proof}

\begin{proof}[\large \textbf{Proof of Theorem \ref{T4}}]
\noindent $(i) \Rightarrow (ii)$ is trivial and $(ii) \Rightarrow (iii)$ follows from Example \ref{Ex2}.

\noindent $(iii) \Rightarrow (i)$:
For any $x\in X$, the set $S(x) = \{s : s\in \RR,\;s=x_{n}\;$ for some $n\in \N\}$ is a non-empty subset of $Y$. 
Since $Y(<)$ is well-ordered,  we can infer that $S(x)\subset Y$ must contain a first element, and so:
\begin{equation*}
W(x) = \min \{x_{n}\}_{n\in \N}
\end{equation*}%
is well-defined.%
\footnote{This SWF satisfies Hammond equity, (see  \citet{alcantud2013}, \citet{dubey2014c}) weak Pareto, monotonicity and anonymity (see \citet{mitra2007c}) as well.}

To verify that $W$ satisfies WE, let $x$, $y\in X$, with $x_{n} < y_{n} < y_{\alpha(n)} < x_{\alpha(n)}$ for all $n\in \dom (\alpha)$.
Let $i\in \N$ be such that $W(x) = x_{i}$.
Then $x_{i} < y_{i} < y_{\alpha(i)} < x_{\alpha(i)}$ must hold true.
We claim that $\min \{y_{n}\}_{n\in \N} > x_i= W(x)$.
If not, then there exist $j\in \N$ such that $y_j\leq x_i$.
But then $\{j, \alpha(j)\}$ must be such that $x_{j} < y_{j} < y_{\alpha(j)} < x_{\alpha(j)}$.
So we get $x_{j} < y_{j} \leq x_{i}$ contradicts the fact that $\min \{x_{n}\}_{n\in \N} = x_i$.
Therefore, $W(y) =\min \{y_{n}\}_{n\in \N} > x_i= W(x)$ as is required.

To verify that $W$ satisfies M, let $x$, $y\in X$ with $y\geq x$.
Then, $W(y) = y_{m}$ for some $m\in \N$, and $y_{m} \geq x_{m}\geq W(x)$, so that $W(y)\geq W(x)$.
Moreover, it is trivial to see that $W$ satisfies AN as well.
\end{proof}

\begin{proof} [\large \textbf{Proof of Theorem \ref{T5}}]
\noindent $(i) \Rightarrow (ii)$ is trivial and $(ii) \Rightarrow (iii)$ follows from Example \ref{Ex1}.

\noindent $(iii) \Rightarrow (i)$:
Let $\underline{Y}=\inf Y$ and $\overline{Y}=\sup Y$.
Consider the following SWF for $x=\left(x_{n}\right)_{n=1}^{\infty }\in X:= Y^\N$%
\begin{equation*}
W(x) = \rho \inf \left\{\left|x_{n}-\underline{Y}\right|\right\}_{n\in \N}+(1-\rho )\inf \left\{\left|\overline{Y}-x_n\right|\right\}_{n\in \N}
\end{equation*}%
where $\rho \in (0, 1)$ is a parameter. 
We show that $W$ satisfies $\we$.
Let $x$, $z\in X$ be such that for all $t\in \N$, there exists $\alpha(t)$ with 
\[
x_t<z_t<z_{\alpha(t)}<x_{\alpha(t)}\;\text{or}\;x_{\alpha(t)}<z_{\alpha(t)}<z_t<x_t.
\]
We claim that $W(z)>W(x)$. 
Clearly, by definition of $W$, we have $W(z)\geq W(x)$ and so if the claim is false, it must be the case that%
\begin{equation}
W(z)=W(x)  \label{8P1}
\end{equation}%
Since%
\begin{equation*}
\inf \left\{\left|z_{n}-\underline{Y}\right|\right\}_{n\in \N}\geq \inf \left\{\left|x_{n}-\underline{Y}\right|\right\}_{n\in \N}\;\text{and}\; \inf \left\{\left|\overline{Y}-z_n\right|\right\}_{n\in \N}\geq \inf \left\{\left|\overline{Y}-x_n\right|\right\}_{n\in \N} \label{8P2}
\end{equation*}%
and $\rho \in (0,1)$, from (\ref{8P1}) it follows that:
\begin{equation*}
a:= \inf \left\{\left|z_{n}-\underline{Y}\right|\right\}_{n\in \N}=\inf \left\{\left|x_{n}-\underline{Y}\right|\right\}_{n\in \N}\;\text{and}\; b:=\inf  \left\{\left|\overline{Y}-z_n\right|\right\}_{n\in \N} =\inf \left\{\left|\overline{Y}-x_n\right|\right\}_{n\in \N}  \label{8P3}
\end{equation*}%
Clearly $a$, $b\in [0,1]$.
We now break up our analysis into the following cases.
\begin{enumerate}[(i)]
\item{$\left\{\left|z_{n}-\underline{Y}\right|\right\}_{n\in \N}$ has a minimum.
Let $k\in \N$ be such that $\left|z_{k}-\underline{Y}\right| = \min \left|z_{n}-\underline{Y}\right|_{n\in \N}$. 
Then, we have
\begin{equation*}
a=\inf \left \{ \left|z_{n}-\underline{Y}\right| \right\}_{n\in \N} =\min \left\{ \left|z_{n}-\underline{Y}\right| \right\}_{n\in \N} = \left|z_{k}-\underline{Y}\right|>\left|x_{k}-\underline{Y}\right|\geq \inf \left \{ \left|x_{n}-\underline{Y}\right| \right\}_{n\in N}= a,
\end{equation*}%
a contradiction.}
\item{$\left\{\left|z_{n}-\underline{Y}\right|\right\}_{n\in \N}$ does not have a minimum.
This case is further subdivided as follows.
\begin{enumerate}
\item[(a)]{$\left\{\left|z_{n}-\underline{Y}\right|\right\}_{n\in \N}$ does not have a minimum, and $\left\{\left|\overline{Y} - x_n\right|\right\}_{n\in \N}$ has a minimum.
Let $s\in \N$ be such that $\left|\overline{Y} - x_s\right|=\min \left\{\left|\overline{Y} - x_n\right|\right\}_{n\in \N}$. 
Then, we have
\begin{equation*}
b= \inf \left\{\left|\overline{Y} - x_n\right|\right\}_{n\in \N} =\min \left\{\left|\overline{Y} - x_n\right|\right\}_{n\in \N} = \left|\overline{Y} - x_s\right| < \left|\overline{Y} - z_s\right| \leq \inf \left\{\left|\overline{Y} - z_n\right|\right\}_{n\in \N} = b
\end{equation*}%
a contradiction.}
\item[(b)]{Neither $\left\{\left|z_{n}-\underline{Y}\right|\right\}_{n\in \N}$  nor $\left\{\left|\overline{Y} - x_n\right|\right\}_{n\in \N}$ has a minimum.
Let $x_1<z_1<z_{\alpha(1)}<x_{\alpha(1)}$.
Then, we can find $x_{\alpha(1)}<x_{n_{1}}<x_{n_{2}}<x_{n_{3}}<\cdots $ with $x_{n_{k}}\in \left(x_{\alpha(1)}, \overline{Y}\right)$ for $k=1, 2, 3, \cdots$, and $x_{n_{k}}\uparrow (\overline{Y}-b)$ as $k\uparrow \infty$. 
Similarly, we can find $z_1> z_{m_{1}}>z_{m_{2}}>z_{m_{3}}>\cdots $ with $z_{m_{r}}\in \left( \underline{Y}, z_1\right)$ for $r=1, 2, 3,\cdots$, and $z_{m_{r}}\downarrow (\underline{Y}-a)$ as $r\uparrow \infty$.
Therefore, we have%
\begin{equation}
(\underline{Y} -a)<\cdots z_{m_{3}}<z_{m_{2}}<z_{m_{1}}^{\prime}<c<x_{n_{1}}<x_{n_{2}}<x_{n_{3}}< \cdots <(\overline{Y} -b). \label{8P4}
\end{equation}%
Consider the set 
\[
Y^{\prime} = \{x_{n_{1}}, x_{n_{2}}, x_{n_{3}},\cdots\} \;\cup \;\{z_{m_{1}}, z_{m_{2}}, x_{m_{3}},\cdots\}.
\]
Clearly, $Y^{\prime}$ is a subset of $Y$ and because of (\ref{8P4}), we note that

\begin{itemize}
\item{the set $Y^{\prime}$ has neither a maximum nor a minimum, and}
\item{for every cut $[Y_{1}^{\prime}, Y_{2}^{\prime}]$ of $Y^{\prime}$, the set $Y_{1}^{\prime}$ has a maximal element and the set $Y_{2}^{\prime}$ has a least element.}
\end{itemize}
Thus, $Y^{\prime}(<)$ is order isomorphic to $\Z$,  a contradiction.}
\end{enumerate}}
\end{enumerate}
Since we are led to a contradiction in cases (i), (ii)(a) and (ii)(b), and these exhaust all logical possibilities, (\ref{8P1}) cannot hold, and our claim that  $W(z)>W(x)$ is established.
\end{proof}


\bibliographystyle{plainnat}
\setlength{\bibsep}{0pt}
\small{\bibliography{APAnonymity}}

\end{document}